\documentclass[runningheads,a4paper]{llncs}

\usepackage{amssymb}
\usepackage{graphicx}

\newcommand{\keywords}[1]{\par\addvspace\baselineskip
\noindent\keywordname\enspace\ignorespaces#1}

\usepackage{amsfonts,url}
\usepackage{latexsym,dsfont}
\usepackage{amsmath,enumerate}
\usepackage{xspace, times}
\usepackage{bm}
\usepackage{ifpdf}
\usepackage{shadow,shadethm,color}
\usepackage{algorithm, algorithmic}
\usepackage{subfigure}
\usepackage{verbatim}

\usepackage
  [pdfpagelabels,plainpages=false,breaklinks]
  {hyperref}

{\makeatletter
 \gdef\xxxmark{%
   \expandafter\ifx\csname @mpargs\endcsname\relax 
     \expandafter\ifx\csname @captype\endcsname\relax 
       \marginpar{xxx}
     \else
       xxx 
     \fi
   \else
     xxx 
   \fi}
 \gdef\xxx{\@ifnextchar[\xxx@lab\xxx@nolab}
 \long\gdef\xxx@lab[#1]#2{{\bf [\xxxmark #2 ---{\sc #1}]}}
 \long\gdef\xxx@nolab#1{{\bf [\xxxmark #1]}}
 \long\gdef\xxx@lab[#1]#2{}\long\gdef\xxx@nolab#1{}%
}

\let\realbfseries=\bfseries
\def\bfseries{\realbfseries\boldmath}

\begin{document}

\newtheorem{conj}{Conjecture}
\def\figs{.}

\newcommand{\Mb}[1]{\mathbf{#1}}
\newcommand{\Mc}[1]{\mathcal{#1}}
\newcommand{\Ir}{\Mc{I}}
\newcommand{\bs}[1]{\boldsymbol{#1}}
\newcommand{\prob}[1]{{\rm Pr} \left[ #1 \right]}
\newcommand{\expect}[1]{{\rm E} \left[ #1 \right]}

\newcommand{\junk}[1]{}
\newcommand{\ignore}[1]{}

\newcommand{\R}[0]{{\ensuremath{\mathbb{R}}}}
\newcommand{\N}[0]{{\ensuremath{\mathbb{N}}}}
\newcommand{\Z}[0]{{\ensuremath{\mathbb{Z}}}}

\def\floor#1{\lfloor #1 \rfloor}
\def\ceil#1{\lceil #1 \rceil}
\def\seq#1{\langle #1 \rangle}
\def\set#1{\{ #1 \}}
\def\abs#1{\mathopen| #1 \mathclose|}   
\def\norm#1{\mathopen\| #1 \mathclose\|}

\newcommand{\poly}{\operatorname{poly}}
\newcommand{\polylog}{\operatorname{polylog}}

\newcommand{\sse}{\subseteq}
\newcommand{\I}{{\mathcal{I}}}
\newcommand{\M}{{\mathcal{M}}}
\newcommand{\B}{{\mathcal{B}}}
\newcommand{\C}{{\mathcal{C}}}
\newcommand{\Pri}{{\mathcal{P}}}
\newcommand{\Problem}{{\mathcal{P}}}
\newcommand{\e}{\varepsilon}
\newcommand{\eps}{\varepsilon}

\newcommand{\OPT}{\ensuremath{{\sf opt}}}
\newcommand{\ALG}{\ensuremath{{\sf alg}}}
\newcommand{\inarc}{\ensuremath{{\sf in}}}
\newcommand{\outarc}{\ensuremath{{\sf out}}}

\newcommand{\pr}[1]{{\rm Pr} \left[ #1 \right]}
\newcommand{\ex}[1]{{\rm E} \left[ #1 \right]}

\newcommand{\opt}[1]{\ensuremath{\mathsf{Opt(#1)}}\xpace}
\newcommand{\Opt}{\ensuremath{\mathsf{Opt}\xspace}}
\newcommand{\LPOpt}{\ensuremath{\mathsf{LPOpt}\xspace}}
\newcommand{\LP}{\ensuremath{\mathsf{LP}\xspace}}

\newcommand{\alg}{\ensuremath{\mathsf{Alg}}\xspace}
\newcommand{\timestep}{\ensuremath{2}}

\newcommand{\tildeo}{\smash{\widetilde{O}}}
\newcommand{\E}{\mathbb{E}}
\newcommand{\D}{\mathcal{D}}
\newcommand{\F}{{\ensuremath{\mathcal{F}}}}
\newcommand{\s}{\mathcal{S}}

\renewcommand{\theequation}{\thesection.\arabic{equation}}
\renewcommand{\thefigure}{\thesection.\arabic{figure}}

%

\newtheorem{SAlg}{Algorithm}
\newtheorem{model}{Model}

\newcommand{\Set}[1]{\mathds{#1}}
\newcommand{\Rand}[1]{\mathbf{#1}}


\mainmatter  

\title{AdCell: Ad Allocation in Cellular Networks}

\titlerunning{AdCell: Ad Allocation in Cellular Networks}

%
%
 \author{
 Saeed Alaei\footnote{Supported in part by NSF Grant CCF-0728839.}\inst{1}
 \and Mohammad T. Hajiaghayi\footnote{Supported in part by NSF CAREER Award, ONR Young Investigator Award, and Google Faculty Research Award.}\inst{1}\inst{2}
 \and Vahid Liaghat\footnote{Supported in part by NSF CAREER Award, ONR Young Investigator Award, and Google Faculty Research Award.}\inst{1}
 \and \\ Dan Pei\inst{2}
 \and Barna Saha\footnote{Supported in part by NSF Award CCF-0728839, NSF Award CCF-0937865.}\inst{1}
 \\
 \email{\{saeed, hajiagha, vliaghat, barna\}@cs.umd.edu, peidan@research.att.com}
 }
\authorrunning{S. Alaei, M.T. Hajiaghayi, V. Liaghat, D. Pei and B. Saha}

       \institute{University of Maryland, College Park, MD, 20742  \\ \and
       AT\&T Labs - Research, 180 Park Avenue, Florham Park, NJ 07932}

%
%

\toctitle{AdCell: Ad Allocation in Cellular Networks}
\tocauthor{AHLPS}
\maketitle

\begin{abstract}
With more than four billion usage of  cellular phones worldwide, mobile advertising has become an attractive
alternative to online advertisements. In this paper, we propose a new targeted advertising policy for Wireless Service
Providers (WSPs) via SMS or MMS- namely {\em AdCell}. In our  model,
a WSP charges the advertisers for showing their ads. Each advertiser has a valuation for specific types of customers in
various times and locations and has a limit on the maximum available budget. Each query is in the form of time and
location and is associated with one individual customer. In order to achieve a non-intrusive delivery, only a limited
number of ads can be sent to each customer. Recently, new services have been introduced that offer location-based
advertising over cellular network that fit in our model (e.g., ShopAlerts by AT\&T) .

We consider both online and offline version of the AdCell problem and develop approximation algorithms with constant
competitive ratio. For the online version, we assume that the appearances of the queries follow a stochastic
distribution and thus consider a Bayesian setting. Furthermore, queries may come from different distributions on different times. This model generalizes several previous advertising models such as
online secretary problem~\cite{HKP04}, online bipartite matching~\cite{firstmatch,vahabbeat} and
AdWords~\cite{saberi05}. Since our problem generalizes the well-known secretary problem, no non-trivial approximation
can be guaranteed in the online setting without stochastic assumptions. We propose an online algorithm that is simple,
intuitive and easily implementable in practice. It is based on pre-computing a fractional solution for the expected
scenario and relies on  a novel use  of dynamic programming  to compute the conditional expectations. We give tight
lower bounds on the approximability of some variants of the problem as well. In the offline setting, where
full-information is available, we achieve near-optimal bounds, matching the integrality gap of the considered linear
program. We believe that our proposed solutions can be used for other advertising settings where personalized
advertisement is critical.

\keywords{Mobile Advertisement, AdCell, Online, Matching}
\end{abstract}

\pagebreak
\sloppy

\section{Introduction}

In this paper, we propose a new mobile advertising concept called {\em Adcell}. More than $4$ billion cellular phones
are in use world-wide, and with the increasing popularity of smart phones, mobile advertising holds the prospect of
significant growth in the near future. Some research firms~\cite{cellstatics} estimate mobile advertisements to reach a
 business worth over 10 billion US dollars by 2012. Given the built-in advertisement solutions
 from popular smart phone OSes, such as iAds for Apple's iOS,  mobile advertising market is poised with even faster growth.

In the mobile advertising ecosystem, wireless service providers (WSPs) render the physical delivery
infrastructure, but so far WSPs have been more or less left out from profiting via mobile advertising
because of several challenges.
First, unlike web, search, application, and game providers, WSPs typically do not have users' application context, which makes it difficult to provide targeted advertisements. Deep Packet Inspection (DPI) techniques that examine packet traces in order to understand  application context, is often not an option because of privacy and legislation issues (i.e., Federal Wiretap Act).  Therefore, a targeted advertising solution for WSPs need to utilize \emph{only the information it is allowed to collect by government and by customers via opt-in mechanisms}.
Second, without the luxury of application context, targeted ads from WSPs require \emph{non-intrusive  delivery methods}.
While users are familiar with other ad forms such as banner, search, in-application, and in-game, push ads with no application context (e.g., via SMS) can be intrusive and annoying if not done carefully. The number and frequency of ads both need to be well-controlled.
Third, targeted ads from WSPs  should be well personalized such that the users have incentive to read the advertisements and take purchasing actions, especially given the requirement that the number of ads that can be shown to a customer is limited.

In this paper, we propose a new mobile targeted advertising strategy, \emph{AdCell}, for WSPs that deals with the above challenges. It takes advantage of the detailed real-time location information of users. Location can be tracked upon users' consent. This is already being done in some services offered by WSPs, such as Sprint's Family Location and AT\&T's Family Map, thus there is no associated privacy or legal complications.
To locate a cellular phone, it must emit a roaming signal to contact some nearby antenna tower, but the process does not require an active call. GSM localization is then done by multi-lateration\footnote{The process of locating an object by accurately computing the time difference of arrival of a signal emitted from that object to three or more receivers.} based on the signal strength to nearby antenna masts~\cite{LBStech}.
Location-based advertisement is not completely new. Foursquare mobile application allows users to explicitly "check in" at places
such as bars and restaurants, and the shops can advertise accordingly.
Similarly there are also automatic proximity-based advertisements using GPS or bluetooth.
For example, some GPS models from Garmin display ads for the nearby business based on the GPS locations~\cite{garmin}.
ShopAlerts by AT\&T \footnote{http://shopalerts.att.com/sho/att/index.html} is another application along the same line.
On the advertiser side, popular stores such as Starbucks are reported to have attracted significant footfalls via mobile coupons.

Most of the existing mobile advertising models are On-Demand, however, AdCell sends the ads via SMS, MMS, or similar methods without any prior notice. Thus to deal with the non-intrusive delivery challenge, we propose user subscription to advertising services that deliver only a {\em fixed number} of ads per month to its subscribers (as it is the case in AT\&T ShopAlerts). The constraint of delivering limited number of ads to each customer adds the
  main algorithmic challenge in the AdCell model (details in Section \ref{sec:model}). In order to overcome the incentive challenge, the WSP can ``pay'' users to read ads and purchase based on them through a reward program in the form of credit for monthly wireless bill.
To begin with, both customers and advertisers should sign-up for the AdCell-service provided by the WSP (e.g.,
currently there are 9 chain-companies participating in ShopAlerts). Customers enrolled for the service should sign an
agreement that their {\em location} information will be tracked; but solely for the advertisement purpose. Advertisers
(e.g., stores) provide their advertisements and a maximum chargeable budget to the WSP.
The WSP selects proper ads (these, for example, may depend on time and distance of a customer from a store) and sends them (via SMS)
to the customers. The WSP charges the advertisers for showing their ads and also for successful ads. An ad is deemed successful
if a customer visits the advertised store. Depending on the service plan, customers are entitled to receive different
number of advertisements per month. Several logistics need to be employed to improve AdCell experience and enthuse customers
into participation. We provide more details about these logistics in the full paper.

\subsection{AdCell Model \& Problem Formulation}
\label{sec:model}%


In the AdCell model, advertisers bid for individual customers based on their location and time.
The triple $(k,\ell,t)$ where $k$ is a customer, $\ell$ is a neighborhood
(location) and $t$ is a time forms a {\em query} and there is a bid amount (possibly zero) associated with each query
for each advertiser. This definition of query allows advertisers to customize their bids
based on customers, neighborhoods and time. We assume a customer can only be in one neighborhood at any
particular time and thus at any time $t$ and for each customer $k$, the queries $(k,\ell_{1},t)$ and
 $(k, \ell_{2},t)$ are mutually exclusive, for all distinct $l_1, l_2$. Neighborhoods are places of interest such as shopping malls, airports, etc.
 We assume that queries are generated at certain times (e.g., every half hour) and only
 if a customer stays within a neighborhood for a specified minimum amount of time. The formal problem definition
 of {\em AdCell Allocation} is as follows:

\paragraph*{{\bf AdCell Allocation}} {\em There are $m$ advertisers, $n$ queries and $s$ customers. Advertiser $i$ has a total budget
$b_i$ and bids $u_{ij}$ for each query $j$. Furthermore, for each customer $k \in [s]$, let
$S_k$ denote the queries corresponding to customer $k$ and $c_k$ denote the maximum number of ads 
which can be sent to customer $k$. The capacity $c_k$ is associated with customer $k$ and is dictated by the AdCell
plan the customer has signed up for. Advertiser $i$ pays $u_{ij}$ if his advertisement is shown for query $j$ and if
his budget is not exceeded. That is, if $x_{ij}$ is an indicator variable set to $1$, when advertisement for advertiser
$i$ is shown on query $j$, then advertiser $i$ pays a total amount of $\min(\sum_j  x_{ij} u_{ij}, b_i)$. The goal of
AdCell Allocation is to specify an advertisement allocation plan
such that the total payment $\sum_{i}\min(\sum_j  x_{ij} u_{ij}, b_i)$
is maximized.}

The AdCell problem is a generalization of  the budgeted AdWords allocation problem \cite{chakra:siam10,Srinivasan:2008}
with capacity constraint on each customer and thus is NP-hard. Along with the offline version of the problem, we also
consider its online version where queries arrive online and a decision to assign a query to an advertiser has to be
done right away. With arbitrary queries/bids and optimizing for the worst case, one cannot obtain any approximation
algorithm with ratio better than $\frac{1}{n}$. This follows from the observation that online AdCell problem also
generalizes the { \em secretary problem} for which no deterministic or randomized online algorithm can get
approximation ratio better than $\frac{1}{n}$ in the worst case.\footnote{ The reduction of the {\em secretary problem}
to AdCell problem is as follows: consider a single advertiser with large enough budget and a single customer with a
capacity of $1$. The queries correspond to secretaries and the bids correspond to the values of the secretaries. So we
can only allocate one query to the advertiser.}. Therefore, we consider a stochastic setting.

For the online AdCell problem, we assume that each query $j$ arrives with probability $p_j$. Upon arrival, each query
has to be either allocated or discarded right away. We note that each query encodes a customer id, a location id and a
time stamp. Also associated with each query, there is a probability, and a vector consisting of the bids for all
advertisers for that query. Furthermore, we assume that all queries with different arrival times or from different
customers are independent, however queries from the same customer with the same arrival time are mutually exclusive
(i.e., a customer cannot be in multiple locations at the same time).

\subsection{Our Results and Techniques}
\label{sec:results}%

Here we provide a summary of our results and techniques. We consider both the offline and
online version of the problem. In the offline version, we assume that we know exactly which queries arrive.
In the online version, we only know the arrival probabilities of queries (i.e., $p_1,\cdots,p_m$).

We can write the AdCell problem as the following random integer program in which $\Rand{I}_j$ is the
indicator random variable which is $1$ if query $j$ arrives and $0$ otherwise:
\begin{align}
    & \text{maximize.}                  & &\sum_i \min(\sum_j  \Rand{X}_{ij} u_{ij}, b_i) \tag{$IP_{BC}$} && \label{prog:IP_BC} \\
    & \forall j \in [n]:                & &\sum_i \Rand{X}_{ij} \le \Rand{I}_j \tag{$F$} \label{eq:IP_BC:F} \\
    & \forall k \in [s]:                & &\sum_{j\in S_k} \sum_i \Rand{X}_{ij} \le c_k \tag{$C$} \label{eq:IP_BC:C} \\
    &                                   & &\Rand{X}_{ij} \in \{0,1\} \notag
\end{align}
We will refer to the variant of the problem explained above as $IP_{BC}$. We also consider variants in which there are
either budget constraints or capacity constraints but not both. We refer to these variants as $IP_B$ and $IP_C$
respectively. The above integer program can be relaxed to obtain a linear program $LP_{BC}$, where we maximize $\sum_i
\sum_j  \Rand{X}_{ij} u_{ij}$ with the constraints ($F$), ($C$) and additional budget constraint $\sum_j \Rand{X}_{ij}
u_{ij} \le b_i$ which we refer to by $(B)$. We relax $\Rand{X}_{ij} \in \{0,1\}$ to $\Rand{X}_{ij} \in [0,1]$.
\begin{table}[b]
\centering
\begin{tabular}{|l|l|}
  \hline
  \hline
  {\bf Offline Version} & {\bf Online Version} \\
  \hline \rule{0pt}{3ex}
  \begin{minipage}[t]{0.48\linewidth}
  \begin{itemize}\itemsep8pt
  \item A $3\over 4$-approximation algorithm.
  \item A $4-\epsilon \over 4$-approximation algorithm when $\forall_i \max_j u_{ij}\leq \epsilon b_i$.
  \end{itemize}
  \end{minipage}
  &
  \begin{minipage}[t]{0.48\linewidth}
  \begin{itemize}\itemsep5pt
  \item A $\left(\frac{1}{2}-\frac{1}{e}\right)$-approximation algorithm.
  \item A $\left(1-\frac{1}{e}\right)$-approximation algorithm with only budget constraints.
  \item A $1 \over 2$-approximation algorithm with only capacity constraints.
  \end{itemize}
  \end{minipage}
  \\
  \hline
  \hline
    \end{tabular}
    \caption{Summary of Our Results\label{tab:res}}
  \label{table:summary}
\end{table}
We also refer to the variant of this linear program with only either constraints of type ($B$) or constraints of type
($C$) as $LP_B$ and $LP_C$.

In the offline version, for all $i \in [m]$ and $j \in [n]$, the values of $\Rand{I}_{j}$ are precisely known. For the
online version, we assume to know the $E[\Rand{I}_{j}]$ in advance and we learn the actual value of $\Rand{I}_{j}$
online. We note a crucial difference between our model and the i.i.d model. In i.i.d model the probability of the arrival of a query is independent of the time, i.e., queries arrive from the same distribution on each time. However, in AdCell model a query encodes time (in addition to location and customer id), hence we may have a different distribution on each time. This implies a prophet inequality setting in which on each time, an onlooker has to decide according to a given value where this value may come from a different distribution on different times (e.g. see \cite{Krengel77,HKS07}).

A summary of our results are shown in~\autoref{tab:res}. In the online version, we compare the expected revenue
of our solution with the expected revenue of the optimal offline algorithm. We should emphasis that we make no
assumptions about bid to budget ratios (e.g., bids could be as large as budgets). In the offline version, our result
matches the known bounds on the integrality gap.

We now briefly describe our main techniques.

\medskip \noindent \textbf{Breaking into smaller sub-problems that can be optimally solved using conditional expectation.}
Theoretically, ignoring the computational issues, any online stochastic optimization problem can be solved optimally
using conditional expectation as follows: At any time a decision needs to be made, compute the total expected objective
conditioned on each possible decision, then chose the one with the highest total expectation. These conditional
expectations can be computed by backward induction, possibly using a dynamic program. However for most problems,
including the AdCell problem, the size of this dynamic program is exponential which makes it impractical. We avoid this
issue by using a randomized strategy to break the problem into smaller subproblems such that each subproblem can be
solved by a quadratic dynamic program.

\medskip \noindent \textbf{Using an LP to analyze the performance of an optimal online algorithm against an optimal offline fractional
solution.} Note that we compare the expected objective value of our algorithm against the expected objective value of
the optimal offline fractional solution. Therefore for each subproblem, even though we use an optimal online algorithm,
we still need to compare its expected objective value against the expected objective value of the optimal offline
solution for that subproblem. Basically, we need to compare the expected objective of an stochastic online algorithm,
which works by maximizing conditional expectation at each step, against the expected objective value of its optimal
offline solution. To do this, we create a minimization linear program that encodes the dynamic program and whose
optimal objective is the minimum ratio of the expected objective value of the online algorithm to the expected
objective value of the optimal offline solution. We then prove a lower bound of $\frac{1}{2}$ on the objective value of
this linear program by constructing a feasible solution for its dual obtaining an objective value of $\frac{1}{2}$.

\medskip \noindent \textbf{Rounding method of \cite{saha:ics10} and handling hard capacities.} Handling ``hard capacities'',
those that cannot be violated, is generally tricky in various settings including facility location and many covering
problems \cite{Chuzhoy:2002,Gandhi:2003,Pal:2001}. The AdCell problem is a generalization of the budgeted AdWords
allocation problem with hard capacities on queries involving each customer. Our essential idea is to iteratively round
the fractional LP solution to an integral one based on the current LP structure. The algorithm uses the rounding
technique of \cite{saha:ics10} and is significantly harder than its uncapacitated version.

Due to the interest of the space we differ the omitted proofs to the full paper.

\section{Related Work}

Online advertising alongside search results is a multi-billion dollar business~\cite{lahaie} and is a major source of
revenue for search engines like Google, Yahoo and Bing. A related ad allocation problem is the AdWords
assignment problem~\cite{saberi05} that was motivated by sponsored search auctions. When modeled as an online bipartite
assignment problem, each edge has a weight, and there is a budget on each advertiser representing the upper bound on
the total weight of edges that might be assigned to it. In the offline setting, this problem is NP-Hard, and several
approximations have been proposed~\cite{adwordAzar,Andelman04,chakra:siam10,Srinivasan:2008}. For the online setting, it is
typical to assume that edge weights (i.e., bids) are much smaller than the budgets, in which case there exists a $(1 -
1/e)$-competitive online algorithm~\cite{saberi05}. Recently, Devanur and Hayes~\cite{Devenur09} improved the competitive ratio to $(1
- \epsilon)$ in the stochastic case where the sequence of arrivals is a random permutation.

Another related problem is the online bipartite matching problem which is introduced by Karp, Vazirani, and
Vazirani~\cite{firstmatch}. They proved that a simple randomized online algorithm achieves a $(1 - 1/e)$-competitive
ratio and this factor is the best possible. Online bipartite matching has been considered under stochastic assumptions
in~\cite{randmatch,vahabbeat,shayan10}, where improvements over  $(1 - 1/e)$ approximation factor have been shown. The
most recent of of them is the work of Manshadi et al.~\cite{shayan10} that presents an online algorithm with a
competitive ratio of $0.702$. They also show that no online algorithm can achieve a competitive ratio better than
$0.823$. More recently, Mahdian et al.\cite{mahdian11} and Mehta et al.\cite{mehta11} improved the competitive ratio
to 0.696 for unknown distributions.

\section{Online Setting}
\label{sec:saeed}%

In this section, we present three online algorithms for the three variants of the problem mentioned in the pervious
section (i.e., $IP_{B}$, $IP_{C}$ and $IP_{BC}$).

First, we present the following lemma which provides a means of computing an upper bound on the expected
revenue of any algorithm (both online and offline) for the AdCell problem.

\begin{lemma}[Expectation Linear Program]
Consider a general random linear program in which $\mathfrak{b}$ is a vector of random variables:
\begin{align*}
    & \text{(Random LP)} \\
    & \text{maximize.}                  & &c^T x \\
    & \text{s.t.}                       & &A x \le \mathfrak{b};\ \ x \ge 0 \\
\end{align*}
\vspace{-1 cm}

Let $OPT(\mathfrak{b})$ denote the optimal value of this program as a function of the random variables. Now
consider the following linear program:
\begin{align*}
    &\text{(Expectation LP)} \\
    & \text{maximize.}                  & &c^T x  \\
    & \text{s.t.}                       & &A x \le E[\mathfrak{b}];\ \ x \ge 0 \\
\end{align*}
\vspace{-1 cm}

We refer to this as the \emph{``Expectation Linear Program''} corresponding to the \emph{``Random Linear
Program''}. Let $\overline{OPT}$ denote the optimal value value of this program. Assuming that the original
linear program is feasible for all possible draws of the random variables, it always holds that
$E[OPT(\mathfrak{b})] \le \overline{OPT}$.
\end{lemma}
\begin{proof}
Let $x^*(\mathfrak{b})$ denote the optimal assignment as a function of $\mathfrak{b}$. Since the random LP is
feasible for all realizations of $\mathfrak{b}$, we have $A x^*(\mathfrak{b}) \le \mathfrak{b}$. Taking the
expectation from both sides, we get $A E[x^*(\mathfrak{b})] \le E[\mathfrak{b}]$. So, by setting $x =
E[x^*(\mathfrak{b})]$ we get a feasible solution for the expectation LP. Furthermore, the objective value
resulting from this assignment is equal to the expected optimal value of the random LP. The optimal value of
the expectation LP might however be higher so its optimal value is an upper bound on the expected optimal
value of random LP.
\end{proof}

As we will see next, not only does the expectation LP provide an upper bound on the expected revenue, it also
leads to a good approximate algorithm for the online allocation as we explain in the following online
allocation algorithm. We adopt the notation of using an overline to denote the expectation linear program
corresponding to a random linear program (e.g. $\overline{LP}_{BC}$ for $LP_{BC}$). Next we present an online
algorithm for the variant of the problem in which there are only budget constrains but not capacity
constraints.

\begin{SAlg}
\label{alg:IP_B}%
\textsc{(Stochastic Online Allocator for $IP_B$)}\\
\vspace{-0.4 cm}
\begin{itemize}
\item
Compute an optimal assignment for the corresponding expectation LP (i.e. $\overline{LP}_B$). Let $x^*_{ij}$
denote this assignment. Note that $x^*_{ij}$ might be a fractional assignment.

\item
If query $j$ arrives, for each $i \in [m]$ allocate the query to advertiser $i$ with probability
$\frac{x^*_{ij}}{p_j}$.
\end{itemize}
\end{SAlg}

\begin{theorem}
\label{thm:IP_B}%
The expected revenue of \autoref{alg:IP_B} is at least $1-\frac{1}{e}$ of the optimal value of the
expectation LP (i.e., $\overline{LP}_B$) which implies that the expected revenue of \autoref{alg:IP_B} it is
at least $1-\frac{1}{e}$ of the expected revenue of the optimal offline allocation too. Note that this result
holds even if $u_{ij}$'s are not small compared to $b_i$. Furthermore, this result holds even if we relax the
independence requirement in the original problem and require negative correlation
instead.
\end{theorem}

Note that allowing negative correlation instead of independence makes the above model much more general than
it may seem at first. For example, suppose there is a query that may arrive at several different times but
may only arrive at most once or only a limited number of times, we can model this by creating a new query for
each possible instance of the original query. These new copies are however negatively correlated. We define the negative correlation as follows: 

\begin{definition}[Negative Correlation]
\label{def:nc}%
Let $\Rand{X}_1, \cdots, \Rand{X}_n$ be random variables. For any subset $S \subset \{1, \cdots, n\}$, let
$\Rand{X}_S$ denote the subset of random variables indexed by $S$ and let $X_S$ and $X'_S$ denote two
realization of these random variables. We say that $\Rand{X}_1, \cdots, \Rand{X}_n$ are negatively correlated
iff for any random variable $\Rand{X}_i$ and any subset $\Rand{X}_S$ of random variables (not containing
$\Rand{X}_i$) and any constant $c$, if $X_S \le X'_S$ then $Pr[\Rand{X}_i \le c | \Rand{X}_S = X_S] \le
Pr[\Rand{X}_i \le c | \Rand{X}_S = X_S']$.
\end{definition}

\begin{remark}
It is worth mentioning that there is an integrality gap of $1-\frac{1}{e}$ between the optimal value of the
integral allocation and the optimal value of the expectation LP. So the lower bound of \autoref{thm:IP_B} is
tight. To see this, consider a single advertiser and $n$ queries. Suppose $p_j =\frac{1}{n}$ and $u_{1j}=1$
for all $j$. The optimal value of $\overline{LP}_B$ is $1$ but even the expected optimal revenue of the
offline optimal allocation is $1-\frac{1}{e}$ when $n\to \infty$ because with probability $(1-\frac{1}{n})^n$
no query arrives.
\end{remark}

To prove \autoref{thm:IP_B}, we use the following theorem:
\begin{theorem}
\label{thm:unc}%
Let $C$ be an arbitrary positive number and let $\Rand{X}_1, \cdots, \Rand{X}_n$ be independent random
variables (or negatively correlated) such that $\Rand{X}_i \in [0, C]$. Let
$\mu = E[\sum_i \Rand{X}_i]$. Then:
\begin{center}
$E[\min(\sum_i \Rand{X}_i, C)] \ge (1-\frac{1}{e^{\mu/C}}) C$
\end{center}
Furthermore, if $\mu \le C$ then the right hand side is at least $(1-\frac{1}{e}) \mu$.
\end{theorem}
\begin{proof}[ \autoref{thm:unc}]
\label{proof:thm:unc}%
Define the random variables $\Rand{R}_i = \max(\Rand{R}_{i-1}-\Rand{X}_i, 0)$ and $\Rand{R}_0 = C$. Observe
that for each $i$, $\Rand{R}_i = \max(C-\sum_{j=1}^i \Rand{X}_j, 0)$ so $\min(\sum_{j=1}^i \Rand{X}_j,
C)+\Rand{R}_i= C$. Therefore $E[\min(\sum_{j=1}^i \Rand{X}_j, C)]+E[\Rand{R}_i]=C$ and to prove the theorem
it is enough to show that $E[\Rand{R}_n] \le \frac{1}{e^{\mu/C}}\cdot C$. To show this we will prove the
following inequality:

\begin{align}
    \Rand{R}_i & \le (1-\frac{E[\Rand{X}_i]}{C}) \Rand{R}_{i-1} \label{eq:U} \tag{U}
\end{align}

Assuming that \eqref{eq:U} is true, we can conclude the following which proves the claim.

\begin{align*}
    \Rand{R}_n  & \le C \cdot \prod_{i=1}^n (1-\frac{E[\Rand{X}_i]}{C}) \\
                & \le C \cdot \frac{1}{e^{\mu/C}}
\end{align*}

The last inequality follows from the fact that $\sum_i \frac{E[\Rand{X}_i]}{C} = \frac{\mu}{C}$ and the right
hand side takes its maximum when for all $i$ : $\frac{E[\Rand{X}_i]}{C} = \frac{\mu}{n C}$ and $n \to
\infty$. Furthermore, to prove the second claim, we can use the fact that $(1-x^a) \ge (1-x)a$ for any $a \le
1$ and conclude that $(1-\frac{1}{e^{\mu/C}}) \ge (1-\frac{1}{e}) \frac{\mu}{C} C = (1-\frac{1}{e}) \mu$
whenever $\mu \le C$. Now it only remains to prove the inequality \eqref{eq:U}:

\begin{align*}
    E[\Rand{R}_i]   & = E[\max(\Rand{R}_{i-1}-\Rand{X}_i, 0)] \\
                    & \le E[\max(\Rand{R}_{i-1}-\Rand{X}_i \frac{\Rand{R}_{i-1}}{C}, 0)] \\
                    & = E[\Rand{R}_{i-1}-\Rand{X}_i \frac{\Rand{R}_{i-1}}{C}] \\
                    & = E[\Rand{R}_{i-1}]-\frac{1}{C}E[\Rand{X}_i \Rand{R}_{i-1}] \\
    \intertext{$\Rand{R}_{i-1}$ and $\Rand{X}_i$ are either independent or positively correlated so:}
                    & \le E[\Rand{R}_{i-1}]-\frac{1}{C}E[\Rand{X}_i] E[\Rand{R}_{i-1}]  \\
                    & = (1-\frac{E[\Rand{X}_i]}{C}) E[\Rand{R}_{i-1}]
\end{align*}


That completes the proof.
\end{proof}

Now we prove \autoref{thm:IP_B} using the above theorem.

\begin{proof}[\autoref{thm:IP_B}]
We apply \autoref{thm:unc} to each advertiser $i$ separately. From the perspective of advertiser $i$, each
query is allocated to her with probability $x^*_{ij}$ and by constraint (B) we can argue that
have $\mu = \sum_{j} x^*_{ij} u_{ij} \le b_i = C$ so $\mu \le C$ and by \autoref{thm:unc}, the expected
revenue from advertiser $i$ is at least $(1-\frac{1}{e}) (\sum_{j} x^*_{ij} u_{ij})$. Therefore, overall, we
achieve at least $1-\frac{1}{e}$ of the optimal value of the expectation LP and that completes the proof.
\end{proof}

Next we present an online algorithm for the variant of the problem in which there are only capacity
constrains but not budget constraints.

\begin{SAlg}
\label{alg:IP_C}%
\textsc{(Stochastic Online Allocator for $IP_C$)}\\
\vspace{-0.4 cm}
\begin{itemize}
\item
Compute an optimal assignment for the corresponding expectation LP (i.e. $\overline{LP}_C$). Let $x^*_{ij}$
denote this assignment. Note that $x^*_{ij}$ might be a fractional assignment.

\item
Partition the items to sets $T_1, \cdots, T_u$ in increasing order of their arrival time and such that all of
the items in the same set have the same arrival time.

\item
For each $k \in [s], t \in [u], r \in [c_k]$, let $E_{k,t}^r$ denote the expected revenue of the algorithm
from queries in $S_k$ (i.e., associated with customer $k$) that arrive at or after $T_t$ and assuming that
the remaining capacity of customer $k$ is $r$. We formally define $E_{k,t}^r$ later.

\item
If query $j$ arrives then choose one of the advertisers at random with advertiser $i$ chosen with a
probability of $\frac{x^*_{ij}}{p_j}$. Let $k$ and $T_t$ be respectively the customer and the partition which
query $j$ belongs to. Also, let $r$ be the remaining capacity of customer $k$ (i.e. $r$ is $c_k$ minus the
number of queries from customer $k$ that have been allocated so far). If $u_{ij}+E_{k,t+1}^{r-1} \ge
E_{k,t+1}^r$ then allocate query $j$ to advertiser $i$ otherwise discard query $j$.
\end{itemize}

We can now define $E_{k,t}^r$ recursively as follows:
\begin{align*}
    E_{k,t}^r &= \sum_{j \in T_t} \sum_{i \in [m]} x^*_{ij} \max(u_{ij}+E_{k,t+1}^{r-1}, E_{k,t+1}^{r})\\
                & + (1-\sum_{j \in T_t} \sum_{i \in [m]} x^*_{ij}) E_{k,t+1}^{r} \tag{EXP$_k$} \label{eq:EXP_k}
\end{align*}
Also define $E_{k,t}^0 = 0$ and $E_{k,u+1}^r = 0$. Note that we can efficiently compute $E_{k,t}^r$ using
dynamic programming.

\end{SAlg}

The main difference between \autoref{alg:IP_B} and \autoref{alg:IP_C} is that in the former whenever we
choose an advertiser at random, we always allocate the query to that advertiser (assuming they have enough
budget). However, in the latter, we run a dynamic program for each customer $k$ and once an advertiser is
picked at random, the query is allocated to this advertiser only if doing so increases the expected revenue
associated with customer $k$.

\begin{theorem}
\label{thm:IP_C}%
The expected revenue of \autoref{alg:IP_C} is at least $\frac{1}{2}$ of the optimal value of the expectation
LP (i.e., $\overline{LP}_C$) which implies that the expected revenue of \autoref{alg:IP_C} it is at least
$\frac{1}{2}$ of the expected revenue of the optimal offline allocation for $IP_C$ too.
\end{theorem}

\begin{remark}
 The approximation ratio of \autoref{alg:IP_C} is tight. There is no online algorithm that can achieve
in expectation better than $\frac{1}{2}$ of the revenue of the optimal offline allocation without making
further assumptions. We show this by providing a simple example. Consider an advertiser with a large enough
budget and a single customer with a capacity of $1$ and two queries. The queries arrive independently with
probabilities $p_1=1-\epsilon$ and $p_2 = \epsilon$ with the first query having an earlier arrival time. The
advertiser has submitted the bids $b_{11}=1$ and $b_{12}=\frac{1-\epsilon}{\epsilon}$. Observe that no online
algorithm can get a revenue better than $(1-\epsilon)\times 1 + \epsilon^2 \frac{1-\epsilon}{\epsilon}
\approx 1$ in expectation because at the time query 1 arrives, the online algorithm does not know whether or
not the second query is going to arrive and the expected revenue from the second query is just $1-\epsilon$.
However, the optimal offline solution would allocate the second query if it arrives and otherwise would
allocate the first query so its revenue is $\epsilon \frac{1-\epsilon}{\epsilon} + (1-\epsilon)^2\times 1
\approx 2$ in expectation.
\end{remark}

Next, we show that an algorithm similar to the previous one can be used when there are both budget
constraints and capacity constraints.

\begin{SAlg}
\label{alg:IP_BC}%
\textsc{(Stochastic Online Allocator for $IP_{BC}$)}\\
Run the same algorithm as in \autoref{alg:IP_C} except that now $x^*_{ij}$ is a fractional solution of
$\overline{LP}_{BC}$ instead of $\overline{LP}_{C}$.
\end{SAlg}

\begin{theorem}
\label{thm:IP_BC}%
The expected revenue of \autoref{alg:IP_BC} is at least $\frac{1}{2}-\frac{1}{e}$ of the optimal value of the
expectation LP (i.e., $\overline{LP}_{BC}$) which implies that the expected revenue of \autoref{alg:IP_BC} it
is at least $\frac{1}{2}-\frac{1}{e}$ of the expected revenue of the optimal offline allocation too.
\end{theorem}

Before we prove the last two theorems, we define a simple stochastic knapsack problem which will be used as a
building block in the proof of \autoref{thm:IP_C}.

\begin{definition}[Stochastic Uniform Knapsack]
\label{def:SUK}%
There is a knapsack of capacity $C$ and a sequence of $n$ possible items. Each item $j$ is of size $1$, has a
value of $v_j$ and arrives with probability $p_j$. Let $\Rand{I}_j$ denote the indicator random variable
indicating the arrival of item $j$. We assume that items can be partitioned into sets $T_1, \cdots, T_u$
based on their arrival times such that all the items in the same partition have the same arrival time and are
mutually exclusive (i.e. at most one of them arrives) and items from different partitions are independent.
Furthermore, we assume that $\sum_{j \in [n]} p_j \le C$.
\end{definition}

The following algorithm based on conditional expectation computes the optimal online allocation for this
problem:

\begin{SAlg}
\label{alg:EXP}%
\textsc{(Stochastic Uniform Knapsack - Optimal Online Allocator)}

Consider a stochastic uniform knapsack problem as defined in \autoref{def:SUK}.

\begin{itemize}
\item
For each $t \in [u]$ and $r \in [C]$, let $E_t^r$ denote the expected revenue of the algorithm from queries
that arrive at or after time $t$ (i.e. $T_t, \cdots, T_u$) and assuming that the remaining capacity of the
knapsack is $r$. We formally define $E_t^r$ later.

\item
If item $j$ arrives do the following. Let $t$ be the index of the partition which $j$ belongs to and let $r$
be the remaining capacity of the knapsack. Put item $j$ in the knapsack if $v_j+E_{t+1}^{r-1} \ge E_{t+1}^r$.
\end{itemize}

$E_t^r$ can be defined recursively as follows and can be efficiently computed using dynamic programming:

\begin{align}
    E_{t}^r &= \sum_{j \in T_t} p_j \max(v_{j}+E_{t+1}^{r-1},
                E_{t+1}^{r}) + (1-\sum_{j \in T_t} p_j) E_{t+1}^{r} \tag{EXP} \label{eq:EXP}
\end{align}

Also define $E_t^0=0$ and $E_{u+1}^r=0$.
\end{SAlg}

Clearly the above algorithm achieves the best revenue that any online algorithm can achieve in expectation
for the stochastic uniform knapsack. However, we need a stronger result since we need to compare its revenue
against the optimal value of the expectation LP.

\begin{lemma}
\label{lem:suk}%
Consider the stochastic uniform knapsack problem as defined in \autoref{def:SUK}. Let $O_o$ denote the random
variable representing the expected revenue of \autoref{alg:EXP} for this problem (i.e. $O_o = E_1^C$). Also
define $O_e = \sum_{j} p_j v_j$. Assuming that $\sum_{j } p_j \le C$, the following always holds:

\begin{align*}
    \frac{1}{2} O_e \le E[O_o] \le O_e
\end{align*}
\end{lemma}
\begin{proof}[ \autoref{lem:suk}]
\label{proof:lem:suk}%
The upper bound is trivial. Clearly, no algorithm (offline or online) can get more than $p_j v_j$ revenue in
expectation from each item $j$. So the total expected revenue is upper bounded by $O_e = \sum_{j} p_j v_j$.
Next we prove the lower bound.

To prove the lower bound we first narrow down the instances that would give the smallest
$\frac{E[O_o]}{O_e}$. The plan of the proof is as follows. First, we show that for each $t$ if we replace all
the items arriving at time $t$ (i.e. all items in set $T_t$) with a single item with probability $p_t =
\sum_{j \in T_t} p_j$ and value $v_t= \sum_{j\in T_t} v_j \frac{p_j}{p_t}$, we may only decrease $E[O_o]$ but
$O_e$ does not change. So this replacement may only decrease $\frac{E[O_o]}{O_e}$ and after the replacement,
each partition only contains one item. So, WLOG, we only need to prove the lower bound for instances in which
each partition contains one item. Next, we argue that if we scale all $v_j$'s by a constant, both $E[O_o]$
and $O_e$ are scaled by the same constant. So, WLOG, we assume that $O_e = 1$. Therefore, we only need to
prove a lower bound on the following program:

\begin{align*}
    & \text{minimize.} \quad    & E[O_o] && \\
    & \text{s.t.}               & O_e \ge 1 &&
\end{align*}

We then consider a linear relaxation of the above program and prove a lower bound of $\frac{1}{2}$ on this
relaxation which also implies a lower bound of $\frac{1}{2}$ for the original program. We prove this by
constructing a feasible solution for the dual of this linear program that achieve a value of $\frac{1}{2}$.

In what follows, we explain each step of the proof in more detail:\\
First of all, we claim that if we replace all of the items arriving at time $t$ (i.e., all items in $T_t$)
with a single item with probability $p_t = \sum_{j \in T_t} p_j$ and value $v_t= \sum_{j\in T_t} v_j
\frac{p_j}{p_t}$, then $O_o$ may decrease but $O_e$ is not affected. Let $E[O'_o]$ and $O'_e$ respectively
denote the result of making this replacement. The fact that $O_e$ is not affected is trivial because $v_t p_t
= \sum_{j\in T_t} p_j v_j$ so $O'_e = O_e$. Let ${E'}_{t}^{r}$ denote the expectation after replacing all
items in $T_{t^*}$ with a single item as explained. For all values of $t > t^*$ nothing is affected so
${E'}_k^r = E_k^r$. Consider what happens at $t=t^*$ when we make the replacement:

\begin{align*}
\intertext{From \eqref{eq:EXP} we have:}
    E_{t}^r &= \sum_{j \in T_t} p_j \max(v_{j}+E_{t+1}^{r-1},
                E_{t+1}^{r}) + (1-\sum_{j \in T_t} p_j) E_{h,t+1}^{r} \\
\intertext{for any convex function $f(\cdot)$ and nonnegative $\alpha_i$'s with $\sum_i \alpha_i =1$ it
always holds that $\sum_i\alpha_i f(x_i) \ge f(\sum_i\alpha_i x_i)$ and $\max(x+a,b)$ is a convex function of
$x$ so:}
    E_{t}^r &= p_t \sum_{j \in T_t} \frac{p_j}{p_t} \max(v_{j}+E_{t+1}^{r-1},
                E_{t+1}^{r}) + (1-p_t) E_{h,t+1}^{r} \\
            &\ge p_t \max(\sum_{j \in T_t} \frac{p_j}{p_t} v_{j}+E_{t+1}^{r-1}, E_{t+1}^{r}) + (1-p_t) E_{h,t+1}^{r} \\
            & = p_t \max(v_t+E_{t+1}^{r-1}, E_{t+1}^{r}) + (1-p_t) E_{h,t+1}^{r} \\
            & = {E'}_t^r
\end{align*}

So we proved that ${E'}_{t}^r \le E_t^r$ for $t = t^*$. Furthermore, notice that according to equation
\eqref{eq:EXP}, for each $t$, $E_{t-1}^r$ is an increasing function of $E_{t}^r$ and $E_{t}^{r-1}$ so if
$E_{t}^r$ decreases then $E_{t-1}^r$ may only decrease so for all values of $t \le t^*$ we can argue that
${E'}_t^r \le E_t^r$ and in particular $E[O'_o] = {E'}_1^C \le E_1^C = E[O_o]$. That means the replacement
may only decrease the expected revenue of our algorithm. So if we replace all the items in each $T_t$ with a
single item as explained above one by one we get an instance in which each partition only contains one item
and with a possibly lower expected revenue from our algorithm. Therefore, WLOG, it is enough to prove a lower
bound for the case where each partition contains one item.

Since scaling all  $v_j$'s by a constant scales both $E[O_o]$ and $O_e$ by the same constant, we can scale
all $v_j$'s so that $O_e = 1$. So, WLOG, we only need to prove the lower bound for cases where $O_e=1$. Now,
we argue that the optimal value of the following program gives a lower bound on $E[O_o]$. Therefore, we only
need to prove the optimal value of this program is bonded below by $\frac{1}{2}$.

\begin{align*}
    & \text{minimize.} \quad    & E[O_o]  && \\
    & \text{s.t.}               & O_e \ge 1 &&
\end{align*}

We now rewrite the the previous program as the following linear program with variables $\bm{E}_t^r$ and
$\bm{v}_t$ (with $t \in [u]$ and $r \in [C]$ by using the definition of $E_t^r$ from \eqref{eq:EXP}. Note
that $E[O_o] = E_1^C$. Also note that, in the following, we address each item by the index of the partition
to which it belongs.

\begin{align*}
    & \text{minimize.}                  & & \quad \quad \bm{E}_1^C && \notag \\
    & \forall t \in [u-1], \forall r \in [C]:
                                        & & \bm{E}_t^r  \ge p_t (\bm{v}_t + \bm{E}_{t+1}^{r-1})+ (1-p_t) \bm{E}_{t+1}^{r}  \\
    & \forall t \in [u-1], \forall r \in [C]:
                                        & & \bm{E}_t^r \ge \bm{E}_{t+1}^r \\
    & \forall r \in [C]:                & & \bm{E}_u^r \ge p_u \bm{v}_u \\
    &                                   & & \sum_t p_t v_t \ge 1 \\
    &                                   & & \bm{v}_t \ge 0, \quad \bm{E}_t^r \ge 0
\end{align*}

Notice that any feasible assignment for the original program is also a feasible assignment for the above
program but not vice versa. So the above program is a linear relaxation of the original program and therefore
its optimal value is a lower bound for the optimal value of the original program. \footnote{It is not hard to
show that any optimal assignment of this linear program is also a feasible optimal assignment for the
original program so the optimal value of the linear program and the original program are in fact equal.}

The above linear program is still not quite easy to analyze, so we consider a looser relaxation as we explain
next. First, it is not hard to show that $E_t^r$ as defined in \eqref{eq:EXP} has decreasing marginal value
in $r$ which implies $E_t^{r-1} \ge \frac{r-1}{r} E_t^r$ (This can be proved by induction on $t$ with the
base case being $t=u$ and then proving for smaller $t$'s. We will prove this formally later). Combining this
with the definition of $E_t^r$ from \eqref{eq:EXP}, we get the following inequality:

\begin{align*}
    E_t^r   & = p_t \max(v_t + E_{t+1}^{r-1}, E_{t+1}^{r}) + (1-p_t) E_{t+1}^{r} \\
            & = \max(p_t v_t + p_t E_{t+1}^{r-1}+ (1-p_t) E_{t+1}^{r}, E_{t+1}^{r}) \\
            & \ge \max(p_t (v_t + \frac{r-1}{r} E_{t+1}^{r})+ (1-p_t) E_{t+1}^{r}, E_{t+1}^{r}) \\
            & = \max(p_t v_t + (1-\frac{p_t}{r}) E_{t+1}^{r}, E_{t+1}^{r})
\end{align*}


Next, we can write the following more relaxed linear program with only variables $\bm{E}_t^C$ and $\bm{v}_t$
(with $t \in [u]$):

\begin{align}
    & \text{minimize.}                  & \quad \quad & \bm{E}_1^C && \notag \\
    & \forall t \in [u-1]:              & & \bm{E}_t^C - p_t \bm{v}_t - (1-\frac{p_t}{r}) E_{t+1}^{r} \ge 0 \tag{$\bm{\alpha}_t$}\label{eq:alpha}\\
    &                                   & & \bm{E}_1^C - p_u v_u \ge 0 \tag{$\bm{\alpha_u}$} \label{eq:alpha_u}\\
    & \forall t \in [u-1]:              & & \bm{E}_t^C - \bm{E}_{t+1}^C \ge 0 \tag{$\bm{\beta}_t$} \label{eq:beta}\\
    &                                   & & \sum_{t=1}^{u} p_t \bm{v}_t \ge 1 \tag{$\bm{\gamma}$} \label{eq:gamma} \\
    &                                   & & \bm{v}_t \ge 0 , \quad \bm{E}_t^C \ge 0 \notag
\end{align}

Next, we show that the optimal value of the above program is bounded below by $\frac{1}{2}$ which implies
that the optimal value of the original program is also bounded below by $\frac{1}{2}$ and that completes the
proof. To do this, we present a feasible assignment for the dual program that obtains an objective value of
at least $\frac{1}{2}$. Note that the objective value of any feasible assignment for the dual program gives a
lower bound on the optimal value of the primal program. The following is the dual program:

\begin{align}
    & \text{maximize.}                  & & \bm{\gamma} && \notag \\
    & \forall t \in [u]:                & & \bm{\gamma} p_t - \bm{\alpha}_t p_t \le 0  \tag{$\bm{v}_t$} \\
    &                                   & & \bm{\alpha}_1 + \bm{\beta}_1 \le 1 \tag{$\bm{E}_1^C$} \\
    & \forall t \in [2 \cdots u-1]:     & & \bm{\alpha}_t + \bm{\beta}_t - (1-\frac{p_{t-1}}{C}) \bm{\alpha}_{t-1} - \bm{\beta}_{t-1} \le 0  \tag{$\bm{E}_t^C$} \\
    &                                   & & - (1-\frac{p_{u-1}}{C}) \bm{\alpha}_{u-1} - \bm{\beta}_{u-1} \le 0 \tag{$\bm{E}_u^C$} \\
    &                                   & & \bm{\alpha}_t \ge 0, \quad \bm{\beta}_t \ge 0, \quad \bm{\gamma} \ge 0 \notag
\end{align}

Now, suppose we set all $\bm{\alpha}_t = \bm{\gamma}$ and $\bm{\beta}_t = \bm{\beta}_{t-1}-\frac{p_{t-1}}{C}
\bm{\gamma}$ for all $t$ except $\bm{\beta}_1 = 1-\bm{\gamma}$. From this assignment, we get $\bm{\beta}_t =
1-\bm{\gamma} - \bm{\gamma} \sum_{k=1}^{t-1}\frac{p_k}{C}$. Observe that we get a feasible solution as long
as all $\bm{\beta}_t$'s resulting from this assignment are non-negative. Furthermore, it is easy to see that
$\bm{\beta}_t > 1 - \gamma -\gamma \frac{\sum_{k=1}^u p_k}{C} =1 - 2\bm{\gamma}$. Therefore, for
$\bm{\gamma}=\frac{1}{2}$, all $\bm{\beta}_t$'s are non-negative and we always get a feasible solution for
the dual with an objective value of $\frac{1}{2}$ which completes the main proof. Next, we present the proof
of our earlier claim that $E_t^{r-1} \ge \frac{r-1}{r} E_t^r$.

We now prove that $E_t^{r} \ge \frac{r}{r+1} E_t^{r+1}$ by induction on $t$ with the base case being $t=u$
which is trivially true because $E_u^r = p_u v_u$ for all $r \ge 1$. Next we assume that our claim holds for
$t+1$ and all values of $r$. We then prove it for $t$ and all values of $r$ as follows:

\begin{align*}
    E_t^r   & = p_t \max(v_t + E_{t+1}^{r-1}, E_{t+1}^{r}) + (1-p_t) E_{t+1}^{r} \\
            & = \max(p_t (v_t + E_{t+1}^{r-1})+ (1-p_t) E_{t+1}^{r}, E_{t+1}^{r}) \\
\intertext{\sl Observe that $\max(a,b) \ge \max((1-\epsilon)a+\epsilon b, b)$ for all $\epsilon \in [0,1]$
so:}
    E_t^r   & \ge \max((1-\epsilon)[p_t (v_t + E_{t+1}^{r-1})+ (1-p_t) E_{t+1}^{r}]\\
            & +\epsilon E_{t+1}^{r}, E_{t+1}^{r}) \\
            & = \max((1-\epsilon) p_t (v_t + E_{t+1}^{r-1}+ \frac{\epsilon}{1-\epsilon} E_{t+1}^{r}) \\
            & + (1-p_t) E_{t+1}^{r}, E_{t+1}^{r}) \\
\intertext{\sl Now by applying the induction hypothesis on $E_{t+1}^{r-1}$ and $E_{t+1}^{r}$ and setting
$\epsilon=\frac{1}{r+1}$:}
    E_t^r   & \ge \max(\frac{r}{r+1} p_t [v_t + \frac{r-1}{r}E_{t+1}^{r}+ \frac{1}{r} E_{t+1}^{r}]\\
            & + (1-p_t) \frac{r}{r+1}E_{t+1}^{r+1}, \frac{r}{r+1}E_{t+1}^{r+1}) \\
            & = \frac{r}{r+1}\max(p_t [v_t + E_{t+1}^{r}]+ (1-p_t) E_{t+1}^{r+1}, E_{t+1}^{r+1}) \\
            & = \frac{r}{r+1} E_t^{r+1}
\end{align*}

So we proved that $E_t^{r} \ge \frac{r}{r+1} E_t^{r+1}$ and that completes the proof.

\end{proof}

Now we can prove the main two theorems using \autoref{lem:suk}.

\begin{proof}[\autoref{thm:IP_C}]
We apply \autoref{lem:suk} to the subset of queries associated with each customer $k$ (i.e. $S_k$)
separately. We may think of this as having a knapsack of capacity $c_k$ for customer $k$. Each pair of
advertiser/query, $(i,j)$ is a knapsack item with value $u_{ij}$. All knapsack items of the form $(i,j)$ with
the same $j$ are mutually exclusive (because at most one advertiser is chosen at random) and they all have
the same arrival time. Therefore, by applying \autoref{lem:suk}, from the knapsack of each customer $k$ we
get at least $\frac{1}{2} (\sum_{j \in S_k} \sum_i x^*_{ij} u_{ij})$ in expectation. So overall, we get
$\frac{1}{2}$ of the optimal value of the expectation LP and that completes the proof.
\end{proof}

\begin{proof}[\autoref{thm:IP_BC}]
The proof is essentially the same as the proof of \autoref{thm:IP_C}. The only difference is that we may also
lose at most a factor of $\frac{1}{e}$ from each advertiser due to going over the budget limit. Note that
this is a gross overestimation because using conditional expectation on each customer may result in
discarding some of the queries which would make it less likely for advertisers to hit their budget limit. So
overall, we get at least $\frac{1}{2}-\frac{1}{e}$ of the optimal value of the expectation LP.
\end{proof}

\section{Offline Setting}
\label{sec:offline}

In the offline setting, we explicitly know  all the queries, that is all the customers, locations, items triplets on
which advertisers put their bids. We want to obtain an allocation of advertisers to queries such that the total
payment obtained from all the advertisers is maximized. Each advertiser pays an amount equal to the minimum of
his budget and the total bid value on all the queries assigned to him. Since, the problem is NP-Hard, we
can only obtain an approximation algorithm achieving revenue close to the optimal. The fractional optimal solution of
$LP_{BC}$ (with explicit values for $\mathcal{I}_{j}, j \in [n]$) acts as an upper bound on the optimal revenue.
We round the fractional optimal solution to a nearby integer solution and establish the following bound.

\begin{theorem}
\label{theorem:approx}
Given a fractional optimal solution for $LP_{BC}$, we can obtain an integral solution for AdCell with budget and capacity constraints
that obtains at least a profit of $\frac{4-\max_{i}{\frac{u_{i,max}}{b_i}}}{4}$ of the profit obtained by optimal fractional allocation
and maintains all the capacity constraints exactly.
\end{theorem}

We note that this approximation ratio is best possible using the considered LP relaxation due to an integrality gap example from
\cite{chakra:siam10}. The problem considered in \cite{chakra:siam10} is an uncapacitated version of the AdCell problem, that is
there is no capacity constraint (C) on the customers. Capacity constraint restricts how many queries/advertisements can be assigned to each customer.
We can represent all the queries associated with each customer as a set; these sets are therefore disjoint and has
integer hard capacities associated with them. Our approximation ratio matches the best known bound from
 \cite{chakra:siam10,Srinivasan:2008} for the uncapacitated case.
 In this section, we give a
 high-level description of the algorithm. We present the detailed description and proof in the next section. Our algorithm is based on applying the rounding technique of \cite{saha:ics10} through several iterations.
 The essential idea of the proposed rounding is to apply a procedure called {\bf Rand-move} to the variables of a
  suitably chosen subset of constraints from the original linear program. These sub-system must be underdetermined to ensure that
  the rounding proceeds without violating any constraint and at least one variable becomes integral.
  The trick lies on choosing a proper sub-system at each step of rounding, which again depends on a detailed case analysis
  of the LP structure.

  Let $y^{*}$ denote the LP optimal solution. We begin by simplifying the assignment given by $y^{*}$.
  Consider a bipartite graph $G(\mathcal{B},\mathcal{I}, E^{*})$ with advertisers $\mathcal{B}$ on one side, queries
  $\mathcal{I}$ on the other
  side and add an edge $(i,j)$  between a advertiser $i$ and query $j$, if $y^{*}_{i,j} \in (0,1)$. That is,
  define $E^{*}=\{(i,j) | \, \, 1 > y^{*}_{i,j} > 0\}$. Our first claim is that $y^{*}$ can be modified
  without affecting the optimal fractional value and the constraints such that $G(\mathcal{B},\mathcal{I}, E^{*})$ is a forest.
  The proof follows from  Claim 2.1  of \cite{chakra:siam10}; we additionally show that such assumption of
  forest structure maintains the capacity constraints.
  \begin{lemma}
  \label{lemma:structure}
  Bipartite graph $G=(\mathcal{B},\mathcal{I}, E^{*})$  induced by the edges $E^*$ can be converted to a forest
  maintaining the optimal objective function value.
  \end{lemma}
  \begin{proof}

Consider the graph $G=(\mathcal{B},\mathcal{I}, E^{*})$ and consider one connected component of it. We will
  argue for each component separately and similarly.

\emph{  Cycle Breaking}:
   Suppose there is a cycle in the chosen component. Since $G$ is bipartite, the cycle has even length.
   Let the cycle be $C=\langle i_1, j_1,i_2,j_2,\ldots,i_{l},j_{l},i_{1}\rangle$, that is consider the cycle
   to start from a advertiser node. Consider a strictly positive value $\alpha$ and consider the following update
   of the $y^{*}$ values over the edges in the cycle $C$. We add $z_{a,b}$ to edge $(a,b)$, where
\begin{enumerate}
\item[R1.] $z_{i_1,j_1}=-\beta$
\item[R2.] If we are at an query node $j_{t}$, $t \in [1,l]$, then $z_{j_t,i_{t+1}}=-z_{i_{t},j_{t}}$
\item[R3.] If we are at a advertiser node $i_{t}$, $t \in [1,l]$, then $z_{i_{t},j_{t}}=-\frac{b_{i_{t},j_{t-1}}z_{j_{t-1},i_{t}}}{b_{i_{t},j_{t}}}$
\end{enumerate}

 $\beta$ is chosen such that after the update, all the variables lie in $[0,1]$ and at least one variable
 gets rounded to $0$ or $1$, thus the cycle is broken. Note that the entire update is a function of $z_{i_{1},j_{1}}$.
 For any query node, its total contribution in (Assign) constraint of LP1 remains unchanged. For any advertiser node, except $i_1$, its
 contribution in (Advertiser) constraint and thus in the objective function remains the same. In addition,
 since the assign constraints remain unaffected, all the capacity constraints are satisfied.
  For advertiser $i_1$, its contribution
 decreases by $z_{i_1,j_1}b_{i_{1},j_{1}}$ and increases by $z_{j_{l},i_{1}}b_{i_{1},j_{l}}=z_{i_1,j_1}b_{i_{1},j_{l}}\frac{b_{i_{2},j_{1}}b_{i_{3},j_{2}}\ldots b_{i_{l-1},j_{l-2}}}{b_{i_{2},j_{2}}b_{i_{3},j_{3}}\ldots b_{i_{l-1},j_{l-1}}}$.\\
 If $b_{i_{1},j_{1}} \leq b_{i_{1},j_{l}}\frac{b_{i_{2},j_{1}}b_{i_{3},j_{2}}\ldots b_{i_{l-1},j_{l-2}}}{b_{i_{2},j_{2}}b_{i_{3},j_{3}}\ldots b_{i_{l-1},j_{l-1}}}$, then instead of adding $z_{j_l,i_{1}}$ on the last edge, we add some $c < z_{j_l,i_{1}}$ such that
 $z_{i_1,j_1}b_{i_{1},j_{1}}=c b_{i_{1},j_{l}}$. Thus, we are able to maintain the objective function exactly.
 The assign constraint on the last query $j_{l}$ can only decrease by this change
 and hence all the capacity constraints are maintained as well.

 Otherwise, $b_{i_{1},j_{1}} > b_{i_{1},j_{l}}\frac{b_{i_{2},j_{1}}b_{i_{3},j_{2}}\ldots b_{i_{l-1},j_{l-2}}}{b_{i_{2},j_{2}}b_{i_{3},j_{3}}\ldots b_{i_{l-1},j_{l-1}}}$. In that case, we traverse the cycle in the reverse order, that is, we start by decreasing on $z_{i_{1},j_{l}}$ first and
 proceed similarly.

\end{proof}

  Once, we have such a forest structure, several cases arise and depending on the cases, we define a suitable sub-system on
  which to apply the rounding technique. There are three major cases.

  (i) There is a tree with two leaf advertiser nodes: in that case, we show that applying our rounding technique only
  diminishes the objective function by little and all constraints are maintained.

 (ii) No tree contains two leaf advertisers, but there is a tree that contains one leaf advertiser: we start with a
 leaf advertiser and construct a path spanning several trees such that we either end up with a combined path with
 advertisers on both side or a query node in one side such that the capacity constraint on the set containing that
 query is not met with equality (non-tight constraint). This is the most nontrivial case and a detailed discussion is given in the next section.

 (iii) No tree contains any leaf advertiser nodes: in that case we again form a combined path spanning several trees
 such that the queries on two ends of the combined path come from sets with non-tight capacity constraints.

\section{The Detailed Description and Proofs of the Offline Algorithm}
\label{appendix:offline}

\subsection{Generic Rounding Scheme }
\label{appendix:rounding}
\paragraph*{ Rounding Scheme of \cite{saha:ics10}}
 Suppose we are given a set of linear constraints $Ax \leq b$, where $A$ is a $m \times n$ real matrix,
 $x \in [0,1]^n$ and $b \in \mathbb{R}^{m}$. We are also given an optimal fractional solution $x \in [0,1]^n$ that
 optimizes a particular objective function say, ``$\max c^T x$'', $c \in \mathbb{R}^n$.
 Our goal is to round the variables in $x$ to $\{0,1\}^n$
 such that the value of the objective function
 remains close to the initial fractional optimal and the constraints
 $Ax \leq b$ are maintained to the extent possible.

 Project $x$ to only those components $x'$ with values in $(0,1)$. Suppose $x' \in (0,1)^{n}$. The components, $x \setminus x'$, which are already rounded have their values fixed. Denote the reduced system by $A'x' \leq b'$, where $A'$ is now a $m \times n'$ real matrix, $x' \in [0,1]^{n'}$ and $b' \in \mathbb{R}^m$. Consider only the {\em tightly} satisfied linearly independent constraints from the system $A'x' \leq b'$. That is, these constraints are satisfied with equality and are linearly independent. Suppose, these subset of constraints are $\hat{A}x'=\hat{b}$.
 We compute a $r \in \mathbb{R}^n$, $r \neq {0}^n$, such that $A'r = 0$, if such a $r$ exists. We know that if the system $\hat{A}x'=\hat{b}$ is underdetermined, that is, have more variables than equations, then the nullspace of $A$ is non-empty and thus computing a nontrivial $r$ is easy. Once, such a $r$ is computed, we consider the following two possible updates:

 \smallskip

 {\bf Rand-Move}:

 \noindent
 $
 {\bf \text{\it Update } \hat{x_{new}}=\hat{x}+\alpha r \text{ \it with probability } \frac{\beta}{\alpha+\beta}
 \text{ \it and; }}\\ {\bf
 \hat{x_{new}}=\hat{x}-\beta r \text{ \it with probability } \frac{\beta}{\alpha+\beta}.}
 $

 Here $\alpha$ and $\beta$ are two nonzero reals such that $\hat{x_{new}} \in [0,1]^{n'}$.
 At least one component after update gets rounded to $0$ or $1$, or
 one more constraint from $A'\setminus{\hat{A}}$ becomes tight. It is easy to verify that such  $\alpha$ and
 $\beta$ always exist. Note that $\expect{\hat{x_{new}}}=\hat{x}$ ({\bf PI}).

 If the system $A'r=0$ does not have any nontrivial solution, then we choose suitable constraints to drop from
 $A'$ and make the system underdetermined.

 The process continues until all the variables are rounded and is guaranteed to terminate in polynomial time.

\subsection{Rounding Algorithm}

  Let $y^{*}$ denote the LP optimal solution. We begin by simplifying the assignment given by $y^{*}$.
  Consider a bipartite graph $G(\mathcal{B},\mathcal{I}, E^{*})$ with advertisers $\mathcal{B}$ on one side, queries
  $\mathcal{I}$ on the other
  side and add an edge $(i,j)$  between a advertiser $i$ and query $j$, if $y^{*}_{i,j} \in (0,1)$. That is,
  define $E^{*}=\{(i,j) | \, \, 1 > y^{*}_{i,j} > 0\}$.
  By Lemma~\ref{lemma:structure}, we know that
  $y^{*}$ can be modified without affecting the value of $\LPOpt$ such that $G(\mathcal{B},\mathcal{I}, E^{*})$ is a forest.



 We now have a collection of trees. There can arise several cases at this stage. For each of these cases,
 we identify a set of linear constraints and apply our {\bf Rand-Move} step on the variables in the chosen
 system of linear constraints. We now specify each of these cases and the system of linear constraints associated
 with that case. For {\bf Rand-Move} to be applicable, we show that our chosen linear system is underdetermined.
 For the correctness proof, we show that all the assign and capacity constraints are maintained. Some advertiser constraints
 may get violated, but in the objective an advertiser $i$ can pay at most $B_{i}$. We show indeed the loss in the
 objective is at most $\frac{1}{4}$ of the optimal objective value. Thus, we obtain a $\frac{3}{4}$-approximation.

 Let $y$ denote the LP solution at this stage. There are three main cases to consider:

 \smallskip

 {\bf Case (i).} There is a tree with two leaf advertiser nodes.

 {\bf Case (ii).} No tree contains two leaf advertisers, but there is a tree that contains one leaf advertiser.

 {\bf Case (iii).} No tree contains any leaf advertiser nodes.


\paragraph*{{\bf Case (i).} There is a tree with two leaf advertiser nodes.}

 Consider the unique path $P$ connecting the two leaf advertisers say $i_0$ and $i_l$.
 Suppose $P=\langle i_0,j_1,i_1,j_2,i_2,\ldots,j_l,i_l$. Define a $x$ variable for
 each edge in the path $P$ that takes values in $[0,1]$. Consider the following system of linear constraints,
  \begin{align}
& x_{i_{t-1},j_{t}}+x_{i_{t},j_{t}}=y_{i_{t-1},j_{t}}+y_{i_{t},j_{t}} & \forall t \in [1,l] \label{eq:case1-1}\\
& x_{i_{t},j_{t}}b_{i_{t},j_{t}}+x_{i_{t},j_{t+1}}b_{i_{t},j_{t+1}}= & \notag \\
& y_{i_{t},j_{t}}b_{i_{t},j_{t}}+y_{i_{t},j_{t+1}}b_{i_{t},j_{t+1}} & \forall t \in [1,l-1]  \label{eq:case1-2}\\
& x \in [0,1]^{2l} \label{eq:case1-3}
\end{align}

 We apply {\bf Rand-Move} on the above linear system.

 \begin{lemma}
 \label{lemma:case1}
 The linear system defined by Equations \ref{eq:case1-1} and \ref{eq:case1-2} is underdetermined, Assign constraints for all queries, Capacity
 constraints for all sets and Bidder constraints for all advertisers except the two leaf advertisers are maintained.
 \end{lemma}
 \begin{proof}
 The number of constraints of type \ref{eq:case1-1} is $l$ and the number of constraints of type \ref{eq:case1-2} is
 $l-1$. However the number of variables is $2l$. Constraint \ref{eq:case1-1} ensures all the assign constraints and hence all the
 capacity constraints are maintained. Constraint \ref{eq:case1-2} ensures all the advertisers maintain their budget except probably the two
 leaf advertisers.
\end{proof}

  \paragraph*{{\bf Case (ii).} No tree contains two leaf advertisers, but there is a tree that contains one leaf advertiser.}

  There are several subcases under it. We first consider four simple subcases.

 \medskip \noindent \textbf{Subcase (1):} {\it There is a maximal path between two queries, where the two queries
 belong to the same set and the set-capacity constraint is non-tight}.

 Since the path is maximal, the queries at the start and the end of the path
 are leaf queries and therefore have non-tight assign constraints. Non-tight naturally implies
 the fact that a constraint is not satisfied by equality.
 Suppose the maximal path is $P=\langle j_1,i_1,\ldots,i_{l-1},j_{l} \rangle$ and let the
 value of the edge-variables associated with this path be
 $\langle y_{i_1,j_1},y_{i_1,j_2},y_{i_2,j_2},\ldots, y_{i_{l-1},j_{l-1}},y_{i_{l-1},j_{l}}\rangle$. These $y$ values are treated as constants.
 Define variables $\langle x_{i_1,j_1},x_{i_1,j_2},x_{i_2,j_2},\ldots,$ $x_{i_{l-1},j_{l-1}},x_{i_{l-1},j_{l}}\rangle$ associated with these edges of $P$. Let $S$ be the set containing the queries $j_1$ and $j_l$. Let the capacity
 of $S$ be $c$. In the current solution, considering the rounded variables as well,
 let the total allocation of queries from the set $S$ be
 be $s+y_{i_{1},j_{1}}+y_{i_{l-1},j_{l}}$. That is, $s$ is the sum of values of the queries in $S$ other than $j_1$ and $j_{l}$.
  Consider the following system of linear constraints:

  \begin{align}
& x_{i_1,j_1} \leq 1, x_{i_{l-1},j_{l}} \leq 1 &  \label{subcase-1:1}\\
& x_{i_{t-1},j_{t}}+x_{i_{t},j_{t}}=y_{i_{t-1},j_{t}}+y_{i_{t},j_{t}} & \forall t \in [2,l-1] \label{subcase-1:2}\\
& x_{i_{t},j_{t}}b_{i_{t},j_{t}}+x_{i_{t},j_{t+1}}b_{i_{t},j_{t+1}} = & \notag \\
& y_{i_{t},j_{t}}b_{i_{t},j_{t}}+y_{i_{t},j_{t+1}}b_{i_{t},j_{t+1}} & \forall t \in [1,l-1] \label{subcase-1:3} \\
& x_{i_1,j_1}+ x_{i_{l-1},j_{l}} \leq s-c \label{subcase-1:4}\\
& x \in [0,1]^{l+1}
\end{align}

We apply {\bf Rand-Move} on the above linear system.

\begin{lemma}
\label{lemma:subcase1}
The linear system defined for Subcase 1 under Case (ii) is underdetermined and {\bf Rand-Move} on it maintains all
the constraints, Assign, Bidder, Capacity, of LP-1.
\end{lemma}
\begin{proof}
  Note that, Constraint (\ref{subcase-1:4}) is non-tight. In addition, Constraint (\ref{subcase-1:1}) implies
  that the leaf queries have non-tight assignment constraint. Now, the number of variables associated with
  the above linear-system is $2(l-1)=2l-2$ and the number of tightly satisfied
linearly independent constraints are $2l-3$. Hence, we can employ {\bf Rand-Move}.

Constraint (\ref{subcase-1:2}) implies the assignment constraint of the non-leaf queries are maintained. Constraint (\ref{subcase-1:3})
implies the budget constraint of the non-leaf advertisers, and therefore all the advertisers considered by this system, are maintained.
The capacities of all the sets in which non-leaf queries
participates are automatically maintained. In addition, Constraint (\ref{subcase-1:4}) implies the capacity constraint of the
set involving the leaf queries are maintained as well.
\end{proof}

\medskip \noindent \textbf{Subcase (2):} {\it There is a maximal path between two queries, where the two queries
 belong to two different sets and both set-capacity constraints are {\em non-tight}}.

 This is almost similar to Case (ii).
 Since the path is maximal, the queries at the start and the end of the path
 are leaf queries and therefore have non-tight assign constraints.
 Suppose the maximal path is $P=\langle j_1,i_1,j_2,i_2,\ldots,j_{l-1},i_{l-1},j_{l} \rangle$ and let the
 value of the edge-variables associated with this path be
 $\langle y_{i_1,j_1},y_{i_1,j_2},y_{i_2,j_2},\ldots,$ \\ $y_{i_{l-1},j_{l-1}},y_{i_{l-1},j_{l}}\rangle$. We
 treat these values as constants here. Define variables $\langle x_{i_1,j_1},x_{i_1,j_2},x_{i_2,j_2},\ldots, x_{i_{l-1},j_{l-1}},x_{i_{l-1},j_{l}}\rangle$ associated with these edges of $P$. The set constraint
 involving the query $j_1$  is non-tight and suppose the total sum
 of the values of the queries (rounded and not rounded) belonging to that set is $s+y_{i_{1},j_{1}}$, while its capacity is
 $c$. Similarly, the set constraint
 involving the query $j_l$  is non-tight and suppose the total sum
 of the values of the queries (rounded and not rounded) belonging to that set is $s'+y_{i_{l-1},j_{l}}$, while its capacity is
 $c'$.  Consider the following linear system.

  \begin{align}
& x_{i_1,j_1} \leq 1, x_{i_{l-1},j_{l}}\leq 1 &  \label{subcase-2:1}
\\
& x_{i_{t-1},j_{t}}+x_{i_{t},j_{t}}=y_{i_{t-1},j_{t}}+y_{i_{t},j_{t}} & \forall t \in [2,l-1] \label{subcase-2:2}\\
& x_{i_{t},j_{t}}b_{i_{t},j_{t}}+x_{i_{t},j_{t+1}}b_{i_{t},j_{t+1}} = & \notag \\
& y_{i_{t},j_{t}}b_{i_{t},j_{t}}+y_{i_{t},j_{t+1}}b_{i_{t},j_{t+1}} & \forall t \in [1,l-1] \label{subcase-2:3} \\
& x_{i_1,j_1} \leq c-s \label{subcase-2:4}\\
& x_{i_{l-1},j_l} \leq c'-s' \label{subcase-2:5}\\
& x \in [0,1]^{l+1}
\end{align}

Note that changes in the linear system from Subcase 1. We apply {\bf Rand-Move} on the above linear system.

\begin{lemma}
\label{lemma:subcase2}
The linear system defined for Subcase 2 under Case (ii) is underdetermined and {\bf Rand-Move} on it maintains all
the constraints, Assign, Bidder, Capacity, of LP-1.
\end{lemma}
\begin{proof}
The constraints (\ref{subcase-2:4}) and (\ref{subcase-2:5}) are non-tight and so are \ref{subcase-2:1}.
The number of variables associated with the above linear-system
is $2(l-1)=2l-2$ and the number of tightly satisfied
linearly independent constraints are $2l-3$. Hence, we employ {\bf Rand-Move}.

 Constraint (\ref{subcase-2:2}) implies the assignment constraint of the non-leaf queries are maintained. Constraint (\ref{subcase-2:3})
implies the budget constraint of the non-leaf advertisers, and therefore all the advertisers considered by this system, are maintained.
The constraints (\ref{subcase-2:4}), (\ref{subcase-2:5}) ensure that we won't violate the capacity constraint of the sets involving the
leaf queries $j_1$ and $j_{l}$ respectively.
\end{proof}

  \medskip \noindent \textbf{Subcase (3):} {\it There is a path (not necessarily maximal path) between two queries, where the two queries
 belong to the same set, the set-capacity constraint is tight but both the queries have non-tight assignment constraints}.

Suppose the path is $P=\langle j_1,i_1,j_2,i_2,\ldots,j_{l-1},i_{l-1},j_{l} \rangle$ and let the
 value of the edge-variables associated with this path be
 $\langle y_{i_1,j_1},$ $y_{i_1,j_2},y_{i_2,j_2},\ldots, y_{i_{l-1},j_{l-1}},y_{i_{l-1},j_{l}}\rangle$. We
 treat these values as constants here. Define variables $\langle x_{i_1,j_1},x_{i_1,j_2},x_{i_2,j_2},\ldots,$ \\ $x_{i_{l-1},j_{l-1}},x_{i_{l-1},j_{l}}\rangle$ associated with these edges of $P$. Let the total fractional assignment
 of query $j_1$ be $a_1+y_{i_{1},j_{1}} < 1$ and the total fractional assignment
 of query $j_l$ be $a_2+y_{i_{l-1},j_{l}} < 1$.  Here we will apply the {\bf Cycle Breaking} trick. We consider
 updates $\langle z_{i_1,j_1},z_{i_1,j_2},z_{i_2,j_2},\ldots, z_{i_{l-1},j_{l-1}},z_{i_{l-1},j_{l}}\rangle$ such that

 \begin{enumerate}
\item[R1.] $z_{j_1,i_1}=-\beta$
\item[R2.] If we are at an query node $j_{t}$, $t \in [1,l]$, then \\ $z_{j_t,i_{t+1}}=-z_{i_{t},j_{t}}$
\item[R3.] If we are at a advertiser node $i_{t}$, $t \in [1,l]$, then \\ $z_{i_{t},j_{t}}=-\frac{b_{i_{t},j_{t-1}}z_{j_{t-1},i_{t}}}{b_{i_{t},j_{t}}}$
\end{enumerate}

The value of $\beta >0$ is so chosen  that ensures
 all the edge-variables remain in $[0,1]$, $x_{i_{l-1},j_{l}} \leq 1-a_2$, $x_{i_{1},j_1} \leq 1-a_1$.
 The entire update is a function of $z_{j_1,i_1}$. If $z_{j_1,i_1} \geq z_{j_l,i_{l-1}}$, then we apply the above update.
 Else, we consider the updates in the reverse direction, starting from the edge $(j_l,i_{l-1})$.

\begin{lemma}
\label{lemma:subcase3}
The update vector ${\bf z}$ is nontrivial and the update maintains all
the constraints, Assign, Bidder, Capacity, of LP-1.
\end{lemma}
\begin{proof}
Clearly, all the advertiser nodes maintain their budget due to rule R3. All the query nodes, except $j_1$ and $j_{l}$ maintain
their assign constraint. All the sets that do not contain $j_1$ or $j_{l}$ thus maintain their capacity constraints.
We start the update, by subtracting from the edge $(j_1,i_1)$ if $z_{j_1,i_1} \geq z_{j_l,i_{l-1}}$. Therefore, the set
that contains both $j_1$ and $j_l$ satisfy its capacity reduced. Otherwise, we start subtracting from the edge $(j_l,i_{l-1}$,
and again the set containing $j_1$ and $j_l$ maintains the capacity constraint, since now  $z_{j_1,i_1} < z_{j_l,i_{l-1}}$.

Since, $y_{i_{1},j_{1}} < 1 -a_1$, $y_{i_{l-1},j_{l}}< 1- a_2$ and all the other variables are in $(0,1)$, we can always find
a $\beta > 0$ such that either $x_{i_{1},j_{1}}=1-a_1$ or $x_{i_{l-1},j_{l}}=1-a_2$, or one of them is rounded down to 0, or some other
variable in the path is rounded to $0$ or $1$.
\end{proof}

 \medskip \noindent \textbf{Subcase (4):} {\it There is a maximal path with a advertiser on one side, an query in another
 with the set containing it being non-tight}.

 Since we are considering a maximal path, the two end-points must be leaf nodes. Suppose the maximal path is $P=\langle j_1,i_1,j_2,i_2,$ $\ldots,j_{l-1},i_{l-1} \rangle$ and let the
 value of the edge-variables associated with this path be
 $\langle y_{i_1,j_1},y_{i_1,j_2},y_{i_2,j_2},\ldots, y_{i_{l-1},j_{l-1}}\rangle$. Let the set in which the
 query $j_{1}$ belongs be $S$ and let it have a total assignment from the rounded and yet to be rounded variables equalling
  $s+y_{i_{l-1},j_{l-1}}$. In addition, let its capacity be $c$. Consider the following linear system:

  \begin{align}
& x_{i_1,j_1} \leq 1 &  \label{subcase-4:1}
\\
& x_{i_{t-1},j_{t}}+x_{i_{t},j_{t}}=y_{i_{t-1},j_{t}}+y_{i_{t},j_{t}} & \forall t \in [2,l-1] \label{subcase-4:2}\\
& x_{i_{t},j_{t}}b_{i_{t},j_{t}}+x_{i_{t},j_{t+1}}b_{i_{t},j_{t+1}}= & \notag \\
& y_{i_{t},j_{t}}b_{i_{t},j_{t}}+y_{i_{t},j_{t+1}}b_{i_{t},j_{t+1}} & \forall t \in [1,l-2] \label{subcase-4:3} \\
& x_{i_1,j_1} \leq c-s \label{subcase-4:4}\\
& x \in [0,1]^{l+1}
\end{align}

 We apply {\bf Rand-Move} on the above linear system.

\begin{lemma}
\label{lemma:subcase4}
The linear system defined for Subcase 4 under Case (ii) is underdetermined and {\bf Rand-Move} on it maintains
the constraints, Assign, Capacity, of LP-1 as well as the Bidder constraint except possibly for the one leaf advertiser.
\end{lemma}
\begin{proof}
The constraints \ref{subcase-4:1} and \ref{subcase-4:4} are non-tight. The number of tightly satisfied linear independent constraints
is therefore at most $(l-2)+(l-2)=2l-4$, whereas the number of variables is $2l-3$. Hence {\bf Rand-Move} can be applied.

Constraint \ref{subcase-4:2} and \ref{subcase-4:1} ensure that all the assign constraints for the queries are maintained.
Constraint \ref{subcase-4:3} ensures the advertiser constraints are maintained for all the advertisers except possibly for $i_{l-1}$.
Constraint \ref{subcase-4:2} and \ref{subcase-4:4} ensure that all the capacity constraints are maintained.
  \end{proof}

 As long as Case (i) or (1-4) subcases of Case (ii) apply, we continue applying them. Also at any time, if we find the linear-system composed of
all the tightly satisfied linearly independent constraints of LP-1 for any tree becomes underdetermined, we apply {\bf Rand-Move}.
When neither subcase (1)-(4) or Case (i) apply, or {\bf Rand-Move} can not be applied to the whole system, we
 have the following properties of the resulting forest structure:

\begin{enumerate}
\item (Case 1 does not apply): No two leaves are advertisers. So there can be at most one leaf that is a advertiser in any tree.
\item (Subcase 3 does not apply): No two queries that are non-tight and belong to the same set with tight capacity are in the same tree. Therefore,
each tree can contain only one non-tight query from a tight set.
\item ({\bf Rand-Move} does not apply to the LP1 constraints for any single tree): The number of tightly satisfied linearly independent constraints from each tree is at least as many as the number of
variables.
\item (Subcase 1 and 2 do not apply): No two leaves that are queries belong to the same set. Also among the leaves that are queries, at most one can belong to a set that has non-tight capacity constraint. In essential,
there can be only one leaf that is an query and that belongs to a set that has non-tight capacity constraint.
\item (Subcase 4 does not apply): If there is a leaf node that is a advertiser in a tree, all other leaf nodes must be queries and must be part of sets
that have tight capacity constraint.
\end{enumerate}

\medskip \noindent \textbf{Subcase (5):} {\it None of subcases (1)-(4) apply}.

 This is the most nontrivial subcase. Denote the tree that contains a leaf advertiser node by $T_1$ and let $i_1$ be the advertiser that is a leaf. Consider a maximal path starting from $i^{1}_1$.
Since Case (i) or Subcases (1-4) do not apply, the other leaf end-point is an query, say $j^{1}_{l_1}$, that belongs to set $S_1$ and set
$S_1$ has tight capacity constraint. Of course, the query $j^{1}_{l_1}$ has non-tight assign constraint since it is a leaf node.
 Let the path be as follows:
\[
P_{1}=\langle i^{1}_1,j^{1}_1,i^{1}_2,j^{1}_2,\ldots,i^{1}_{l_1-1},j^{1}_{l_1-1},i^{1}_{l_1},j^{1}_{l_1}\rangle.
\]

Since subcase 3 does not apply, tree $T_1$ does not contain any other non-tight query from $S_1$. Now capacities
are always integer and set $S_1$ has tight capacity constraint. This implies that set $S_1$ must contain another
non-tight query and that non-tight query must belong to a different tree. Denote this second tree by $T_2$ and
call this another non-tight query of $S_1$ by $j^{2}_{1}$. If $T_2$ contains a leaf node that is a advertiser, consider
the path from $j^{2}_{1}$ to that advertiser node. Say the path is,
\[
P_{2}=\langle j^{2}_{1}, i^{2}_{1},j^{2}_{2},\ldots, i^{2}_{l_2-2}, j^{2}_{l_2-1}, i^{2}_{l_2-1},j^{2}_{l_2},i^{2}_{l_2} \rangle.
\]

Consider a combined path $\langle P_1,P_2 \rangle$.
$$
\langle P_1,P_{2} \rangle= \langle i^{1}_1,j^{1}_1,\ldots, i^{1}_{l_1-1}, j^{1}_{l_1-1},i^{1}_{l_1},\underbrace{j^{1}_{l_1}, j^{2}_{1}}, i^{2}_{1},j^{2}_{2},\ldots, j^{2}_{l_2},i^{2}_{l_2} \rangle .
$$

Essentially this combined path is thought of a single path ending at two leaf advertisers. We apply the rounding of Case (i) in this scenario with a
slight change in handling the job nodes. We rewrite the linear system for convenience.

  \begin{align}
& x_{i^{1}_{t-1},j^{1}_{t}}+x_{i^{1}_{t},j^{1}_{t}}=y_{i^{1}_{t-1},j^{1}_{t}}+y^{1}_{i_{t},j^{1}_{t}}  \forall t \in [1,l_1-1] & \label{subcase5:case1-1}\\
& x_{i^{2}_{t-1},j^{2}_{t}}+x_{i^{2}_{t},j^{2}_{t}}=y_{i^{2}_{t-1},j^{2}_{t}}+y^{2}_{i_{t},j^{2}_{t}}  \forall t \in [2,l_2] & \label{subcase5:case1-2}\\
&x_{i^{1}_{l_1},j^{1}_{l_1}} \leq 1, x_{i^{2}_{1},j^{2}_{1}} \leq 1 \label{subcase5:case1-3}\\
& x_{i^{1}_{l_1},j^{1}_{l_1}} + x_{i^{2}_{1},j^{2}_{1}} \leq y_{i^{1}_{l_1},j^{1}_{l_1}} + y_{i^{2}_{1},j^{2}_{1}} \label{subcase5:case1-4}\\
& x_{i^{a}_{t},j^{a}_{t}}b_{i^{a}_{t},j^{a}_{t}}+x_{i^{a}_{t},j^{a}_{t+1}}b_{i^{a}_{t},j^{a}_{t+1}} = y_{i^{a}_{t},j^{a}_{t}}b_{i^{a}_{t},j^{a}_{t}}+ y_{i^{a}_{t},j^{a}_{t+1}}b_{i^{a}_{t},j^{a}_{t+1}} & \notag \\
& \forall (t,a) \in ([2,l_1],1) \cup ([1,l_2-1],2) &   \label{subcase5:case1-5}\\
& x \in [0,1]^{2l_1+2l_2-2} \label{subcase5:case1-6}
\end{align}

 We apply {\bf Rand-Move} as usual. Note that, essentially we are assuming $j^{1}_{l_1}$ and $j^{2}_{1}$ as a single node
 while writing the constraint \ref{subcase5:case1-4}.
\begin{lemma}
\label{lemma:subcase5}
The linear system defined above in underdetermined and Assign constraints for all queries, advertiser constraints for all advertisers except $i^{1}_1$ and $i^{2}_{l_2}$ and Capacity constraints for all sets are maintained.
\end{lemma}
\begin{proof}
\label{proof:lemma:subcase5}
Again the number of linearly independent tightly satisfied constraints are $(l_1-1)+)(l_2-1)+1+(l_1-1)+(l_2-1)=2l_1+2l_2-3$ from \ref{subcase5:case1-1}, \ref{subcase5:case1-2}, \ref{subcase5:case1-4} and \ref{subcase5:case1-5}. The number of variables is $2l_1+2l_2-2$.
Thus {\bf Rand-Move} can be applied.
From constraints \ref{subcase5:case1-1}, \ref{subcase5:case1-2}, \ref{subcase5:case1-3} we get that all the assign constraints and all the capacity
constraints except for set $S$ are satisfied. Constraint \ref{subcase5:case1-4} ensures that the capacity constraint of the set $S$ is maintained.
Constraint \ref{subcase5:case1-5} maintains all the advertiser constraints except for advertisers $i^{1}_1$ and $i^{2}_{l_2}$.
\end{proof} 

When, the above does not apply, then in $T_2$ there is no leaf node that is a advertiser. If there is a leaf node that is an query but the query is in a set that
has non-tight capacity constraint, then we consider that path $P_2'$ (say) (we use the same symbols as in $P_2$ for $P_2'$, but it is not to be confused with $P_2$, since we are considering $P_2'$ when no such path like $P_2$ exists).
\[
P_{2}'=\langle j^{2}_{1}, i^{2}_{1},j^{2}_{2},\ldots, i^{2}_{l_2}, j^{2}_{l_2}\rangle.
\]

Consider a combined path $\langle P_1,P'_{2}\rangle$ as before, that is we treat $j^{1}_{l_1}$ and $j^{2}_{1}$ as a single node
while maintaining their total contribution to the set $S$. Note because of considering the combined path $\langle P_1,P'_{2}\rangle$, this becomes identical to the subcase 4. So we apply the rounding on this combined path as in subcase 4. The correctness of this rounding step also follows from
Lemma \ref{lemma:subcase4}.

Otherwise, all the leaf nodes in $T_2$ are queries and the sets containing them have tight capacity constraint. Follow a maximal
path from $j^{2}_{1}$ to one such leaf node, say $j^{2}_{l}$, and let it belong to set $S_2$. Denote the maximal path by $P_2''$.

Since subcase 3 does not apply to $T_2$, $T_2$ does not contain another non-tight query from $S_2$. But, the capacity of
$S_2$ is integer and thus it must have another non-tight query. Call that query to be $j^{3}_{1}$ and denote the tree containing
it to be $T_3$. If $T_3$ happens to be same as $T_1$, then consider the path $P'$ in $T_1$ between $j^{3}_{1}$ and $j^{1}_{l_1}$. Now
consider the combined path $\langle P', P_2'' \rangle$. In this combined path the two end-points belong to two non-tight queries from
set $S_2$ that has tight capacity constraint. Thus, this is identical to subcase 3 and we apply the rounding of subcase 3.
The correctness follows again from Lemma \ref{lemma:subcase3}.

Otherwise, $T_3$ is a tree different from both $T_1$ and $T_2$ and we continue similarly from $j^{3}_{1}$. Thus, if at any point of time, we reach
 a leaf node
that is a advertiser or an query in a non-tight set, or an query in a tight-set but for which the another non-tight query belongs to a tree
already visited, we can continue our rounding.

However, it may happen that a tight set contains more than two non-tight queries. In that case, it is
possible to visit a tight set more than twice in our process. So suppose we are at tree $T_g$ and while considering maximal path, $P_i=\langle j^{g}_{1},i^{g}_{1},j^{g}_{2},\ldots,i^{g}_{l_{g}-1},j^{g}_{l_{g}}\rangle$, we get to $j^{g}_{l_{g}}$ that belongs to a set $S^{g}$ that is
already visited. That is, we have already seen two non-tight queries as end-points (one at the end of a maximal path and the other as the start of a maximal path in two consecutive trees) of two maximal paths say in $T_{h}$ and $T_{h+1}$, $h+1 < g$.
Let the maximal paths that have been considered in trees $T_{h+1},T_{h+2},\ldots,T_{g}$ be $P_{h+1},P_{h+2},\ldots,P_{g}$.
Consider the combined path $\langle P_{h+1},P_{h+2},\ldots,P_{g} \rangle$ and note that in this combined path the two end-points belong to two non-tight queries from set $S_g$ that has tight capacity constraint. Thus we apply the rounding of subcase 4. Indeed it is not
required to visit a non-tight query for the third time as an end-point of a maximal path. If at any time in this process, we visit
a third non-tight query from a set with tight capacity constraint, we can write a combined path with two end-points containing non-tight queries from that set and apply rounding of subcase 3.

Otherwise, all the trees visited are different and we keep on continuing this process. Since the number of trees are at most $\min{\{n,m\}}$,
this process must terminate in some tree $T^{t}$ and at some leaf query node $j^{t}_{l_{t}}$ within a tight set $S_{t}$. Since $S_{t}$ has
at least two non-tight queries, the other non-tight query, say $j$, must belong to some tree $T^{t'}, t'< t$. Considering a path from $j$ to
$j^{t'}_{l_{t'}}$ and then following the maximal paths in $T^{t'+1}, T^{t'+2},\ldots, T^{t}$, we again get a combined path on which we can apply
rounding of subcase 3.

\paragraph*{{\bf Case (iii).} No tree contains any leaf advertiser nodes.}

This case is similar to Case (ii). We start with a leaf query, possibly with a leaf query that is in a non-tight set if one exists, and
obtain a combined path on which we can apply one of Subcases (1)-(4).

This completes the description of the rounding method. At every step, the entire rounding procedure takes $poly(n,m)$ time
and at each step we either make a constraint tight or round a variable. Thus we are guaranteed to complete rounding
all the variables to integers in polynomial number of steps.

From the above discussion and Lemma \ref{lemma:case1}-\ref{lemma:subcase5}, we get the following,
\begin{lemma}
The rounding procedure maintains all the assign and the capacity constraints. A advertiser node maintains the advertiser
constraint as long as in the current fractional solution, it is connected to two or more queries with nonzero fractional values.
\end{lemma}

Now, we need to prove that our expected approximation ratio is
$$\frac{4-\max_{i}{\frac{b_{i,max}}{B_i}}}{4}$$
where $b_{i,max}=\max_{j}{b_{i,j}}$. We can always assume $b_{i,max} \leq B_{i}$ without loss of generality for all $i$, we get a $3/4$ approximation. If bids are small, that
is $\max_{i}{\frac{b_{i,max}}{B_i}}\leq \epsilon$, then we get a $(4-\epsilon)/4$ approximation.

\begin{proof}[Theorem~\ref{theorem:approx}]

Let $P_{i}^{0}$ denote the payment made by advertiser $i$ as assigned by LP1. In our rounding process, when an edge-variable
 gets rounded to $0$ or $1$, it is removed permanently or assigned permanently. The forest structure that we consider always
 contains only the fractional edge-variables. If the advertiser $i$ never has degree $1$ in the forest, then by our rounding procedure
 its final payment is same as $P_{i}^{0}$. Therefore, suppose at some stage $s$, advertiser $i$ becomes a leaf node and let $a$ be the so far rounded payment on $i$ and let $b$ be the unique query assigned to advertiser $i$ with fractional assignment $p$ and bid $d$. Note that, all
 $a, b, p, d$ are random variables. If $P_{i}^{s}$ denote the total payment (fractional and integral) done by advertiser $i$ at the
 end of the $s$th iteration, then we have
 \[
  P_{i}^{s}=a+dp=P_{i}^{0}
 \]

 Once a advertiser becomes a leaf node, it only takes part in {\bf Rand-Move}. Let $P_{i}^{s+1}, P_{i}^{s+2},\ldots, P_{i}^{t}$ denote the payment rounded on advertiser $i$ at the end of the iterations $s+1, s+2,\ldots, t$. Assume $t$ is the last iteration. Then we have from property {\bf [P1]} of
 {\bf Rand-Move} that
 \[
 \expect{P_{i}^{g}|P_{i}^{g-1}=a+dp_{g-1}}=a+dp_{g-1}
 \]
 for $g > s$.
 Thus

 $
 \expect{P_{i}^{g}}=\int_{x}\expect{P_{i}^{g}|P_{i}^{g-1}=a+dx}\prob{P_{i}^{g-1}=a+dx}=\\
 \int_{x}a+dx \prob{P_{i}^{g-1}=a+dx}=\expect{P_{i}^{g-1}}.
$

Hence we have
\[
\expect{P_{i}^{t}}=\expect{P_{i}^{t-1}}=\cdots=\expect{P_{i}^{s}}=a+dp=P_{i}^{0}
\]

Then it directly follows from the above,

{\it With probability $1-p$ the rounded payment on advertiser $i$ is $a$ and with probability $p$ the rounded payment is $a+d$},
since $\expect{P_{i}^{t}}=a \prob{\text{edge $(i,b)$ is rounded to $0$}}+ (a+d)\prob{\text{edge $(i,b)$ is rounded to $1$}}$.

Thus the final expected profit from advertiser $i$ is $(1-p)\min{\{B_{i},a\}}+p \min{\{B_{i},a+d\}}$. The profit obtained from
$i$ in the optimal LP solution is $\min{\{B_{i},a+dp\}}$. Therefore, by the linearity of expectation, the expected
approximation ratio is the maximum possible value of
\[
\frac{(1-p)\min{\{B_{i},a\}}+p \min{\{B_{i},a+d\}}}{\min{\{B_{i},a+dp\}}}.
\]

This part of the proof is similar to the analysis of Theorem 1 of \cite{Srinivasan:2008}.
 Let $b_{i,max}=\max_{j}{b_{i,j}}$. We can assume without loss of generality that
$b_{i,max} \leq B_{i}$ for all $i$. It is easy to see that if $a > B_{i}$ or $a+d < B_{i}$, then
the above approximation ratio is $1$. Hence assume, $ a < B_{i} < a+d$. We thus have the approximation
ratio to be
\[
r=\frac{a(1-p)+pB_{i}}{\min{\{B_{i},a+dp\}}}
\]

Now considering the two cases, $B_{i} \leq / > a+dp$, we get the following result:
\[
\frac{(1-p)\min{\{B_{i},a\}}+p \min{\{B_{i},a+d\}}}{\min{\{B_{i},a+dp\}}} \leq \frac{4-\max_{i}{\frac{b_{i,max}}{B_i}}}{4}
\]

Since we can assume without loss of generality $b_{i,max} \leq B_{i}$ for all $i$, we get a $3/4$ approximation. If bids are small, that
is $\max_{i}{\frac{b_{i,max}}{B_i}}\leq \epsilon$, then we get a $(4-\epsilon)/4$ approximation.
\end{proof}

\bibliographystyle{abbrv}
\bibliography{AdCell}

%
%
%
%
%
%
%
%

\end{document}